\documentclass[twoside]{article}
\usepackage{color}
\usepackage{amssymb}
\usepackage{dsfont}
\usepackage{amssymb,amsmath,ulem,cancel,graphicx,enumerate}

\usepackage[sort&compress,square,comma,authoryear]{natbib}

\usepackage[colorlinks=true,urlcolor=blue,citecolor=blue,linkcolor=blue,bookmarks=true]{hyperref}

\numberwithin{equation}{section}

\newtheorem{theorem}{Theorem}[section]
\newtheorem{lem}{Lemma}[section]
\newtheorem{pro}{Proposition}[section]
\newtheorem{cor}{Corollary}[section]
\newtheorem{rem}{Remark}[section]
\newtheorem{rems}{Remarks}[section]
\newtheorem{ex}{Example}[section]
\newtheorem{defi}{Definition}[section]
\newtheorem{hyp}{Assumption}[section]
\newtheorem{con}{Conjecture}[section]

\newcommand{\sect}{\section}
\newcommand{\ssc}{\subsection}

\newcommand{\bt}{\begin{theorem}}
\newcommand{\et}{\end{theorem}}
\newcommand{\bl}{\begin{lem}}
\newcommand{\el}{\end{lem}}
\newcommand{\bp}{\begin{pro}}
\newcommand{\ep}{\end{pro}}
\newcommand{\bcor}{\begin{cor}}
\newcommand{\ecor}{\end{cor}}
\newcommand{\bcon }{\begin{con} \rm }
\newcommand{\econ }{\end{con}}
\newcommand{\lab }{\label }
\newcommand{\bd}{\begin{defi} \rm }
\newcommand{\ed}{\end{defi}}
\newcommand{\brem }{\begin{rem} \rm }
\newcommand{\erem }{\end{rem}}
\newcommand{\brems }{\begin{rems} \rm }
\newcommand{\erems }{\end{rems}}
\newcommand{\bhyp }{\begin{hyp} \rm }
\newcommand{\ehyp }{\end{hyp}}
\newcommand{\bex}{\begin{ex} \rm }
\newcommand{\eex}{\end{ex}}
\newcommand{\be}{\begin{equation}}
\newcommand{\ee}{\end{equation}}
\newcommand{\bde}{\begin{displaymath}}
\newcommand{\ede}{\end{displaymath}}
\newcommand{\beq}{\begin{eqnarray*}}
\newcommand{\eeq}{\end{eqnarray*}}
\newcommand{\beqa}{\begin{eqnarray}}
\newcommand{\eeqa}{\end{eqnarray}}
\newcommand{\bea}{\begin{align*}}
\newcommand{\eea}{\end{align*}}

\def\I{\mathds{1}}

\def\wt{\widetilde}
\def\tilde{\widetilde}
\def\phi{\varphi }

\def\ff{{\mathbb F}}
\def\hh{{\mathbb H}}
\def\gg{{\mathbb G}}

\def\H{{\cal H}}

\def\P{\mathbb P}
\def\Q{\mathbb Q}
\def\EQ{{\mathbb E}_{{\mathbb Q}}}
\def\Qhh{{\mathbb Q}^h}
\def\Ehh{{\mathbb E}^h}
\def\Qhf{{\mathbb Q}^{\beta }}
\def\Ehf{{\mathbb E}^{\beta}}

\newcommand{\Keywords}[1]{\par\noindent{\small{\bf Keywords\/}: #1}}
\newcommand{\Class}[1]{\par\noindent{\small{\bf Mathematics Subjects Classification (2010)\/}: #1}}

\title{{\LARGE \bf 
		Funding, repo and credit inclusive valuation as modified option pricing
		\footnote{The research of Damiano Brigo and Marek Rutkowski was supported by the EPSRC Mathematics Platform Grant EP/I019111/1 {\it Mathematical Analysis of Funding Costs} at Imperial College London.}}\vskip 80 pt}

\author{Damiano Brigo
	\\Dept. of Mathematics \\ Imperial College London \and Cristin Buescu \\ Dept. of Mathematics \\ King's College London \\ \and Marek Rutkowski\footnote{
		The research of Marek Rutkowski was supported by the DVC Bridging Support Grant {\it BSDEs Approach to Models with Funding Costs} at the University of Sydney.} \\ School of Mathematics \\ and Statistics \\ University of Sydney}

\date{\vskip 30 pt \today} 

\begin{document}
\maketitle
\vskip 30 pt
\begin{abstract}
	We take the holistic approach of computing an OTC claim value that incorporates credit and funding liquidity risks and their interplays, instead of forcing individual price adjustments: CVA, DVA, FVA, KVA. The resulting nonlinear mathematical problem features semilinear PDEs and FBSDEs. We show that for the benchmark vulnerable claim there is an analytical solution, and we express it in terms of the Black-Scholes formula with dividends. This allows for a detailed valuation analysis, stress testing and risk analysis via sensitivities. 
\vskip 30 pt
\Keywords{Funding costs, counterparty risk, credit risk, repo market, valuation adjustments, hedging
}
\vskip 20 pt
\Class{91G40,$\,$60J28}
\end{abstract}

\newpage

\sect{Introduction}  \label{sec1}
Prior to the financial crisis of 2007-2008, institutions tended to ignore the credit risk
of highly-rated counterparties in valuing and hedging contingent claims traded
over-the-counter (OTC), claims which are in fact bilateral contracts negotiated between
two default-risky entities. Then, in just the short span of one month of 2008 (Sep 7 to Oct 8),
eight mainstream financial institutions experienced critical credit events in a painful
reminder of the default-riskiness of even large names (the eight were
Fannie Mae, Freddie Mac, Lehman Brothers, Washington Mutual, Landsbanki, Glitnir and Kaupthing, to which we could also add Merrill Lynch).

One of the explosive manifestations
of this crisis was the sudden 
divergence
between the rate of overnight indexed swaps (OISs) and the LIBOR rate, pointing to the
credit and liquidity risk existing in the interbank market. This forced dealers and financial institutions to
reassess the valuation of OTC claims, leading to various adjustments to their book value.

It is difficult to do justice to the entire literature on such valuation adjustments,
which  intertwines two strands that have been developed in parallel by academics and practitioners. For a full introduction to valuation adjustments and all related references we refer to the first chapter of either \citet{BrigoMoriniPallavicini2012} or \citet{CrepeyBook}.  Here we will summarize only the features that are most relevant for this work.

Firstly, the credit valuation adjustment (CVA) corrects the value of a trade with the expected costs borne by
the dealer due to scenarios where the counterparty may default. CVA had been around
for some time (see, for example, \citet{BrigoMas}) and in its most sophisticated version can include credit migration and ratings transition (see, for example, \citet{BR14}). 
Secondly, the debit valuation adjustment (DVA), which 
is simply the CVA seen from the other side of the trade, corrects the price with the expected benefits to the dealer due to scenarios where the dealer might default before the end of the trade. This latter correction may lead to a controversial profit that can be booked when the credit quality of the dealer deteriorates. For example, Citigroup reported in a press release on the first quarter revenues of 2009 that ``Revenues also included [...] a net \$2.5 billion positive CVA on derivative positions, excluding monolines, mainly due to the widening of Citi's CDS spreads".
Accounting standards by the FASB accept DVA, whereas the Basel Committee does not recognize it
in the risk measurement space; see, for example, \citet{BrigoMoriniPallavicini2012} for a detailed discussion. On top of this, DVA is very difficult to hedge, as this would involve selling
protection on oneself. The spread risk is therefore hedged via proxy hedging, trading in names that are thought to be correlated to one's own bank, but this does not help with the jump to default risk; see again  \citet{BrigoMoriniPallavicini2012} for a discussion. In real markets, CVA and DVA are often computed on large netting sets. It may thus be important to study such quantities for large portfolios; see, for example, \citet{bo}.

More recently, the funding valuation adjustment (FVA) was introduced. This is the price adjustment due to the
cost of funding the trade. Trading desks back the deal with a client by hedging it with
other dealers in the market, and this may involve maintaining a number of hedging accounts in the underlying assets, in cash,  or in other correlated assets when proxy-hedging.  Typically, the funds needed for these operations are raised from the internal treasury of the dealer, and ultimately come from external funders. Interest charges on all the related borrowing and lending activities need to be covered, and this is added to the valuation. Michael Rapoport reported in the Wall Street Journal, on Jan 14, 2014, that funding valuation adjustments cost J.P. Morgan Chase \$1.5 billion in the fourth quarter results. More recently, the capital valuation adjustment (KVA) has started being discussed for the cost of capital one has to set aside in order to be able to trade. We will not address KVA here, since even its very definition is currently subject to intense debate in the industry.

All such adjustments may concern both over the counter (OTC) derivatives trades and derivatives trades done through central 
clearing houses (CCP), see for example \citet{BP14} for a comparison of the two cases where the full mathematical structure of the problem of valuation under possibly asymmetric initial and variation margins, funding costs, liquidation delay and credit gap risk is explored, resulting in BSDEs and semilinear PDEs.
It is worth pointing out that the size of such derivatives markets remains quite relevant even post-crisis. 
At end of 2012, the market value of outstanding OTC derivative contracts  was reported to be \$24.7 trillion with \$632.6 trillion in notional value (BIS 2013). Even if many deals are now collateralized in an attempt to avoid CVA altogether, contagion and gap risk may still result in 
important residual CVA, as was shown for the case of credit default swap trades in \citet{BrigoMoriniPallavicini2012}.

As we mentioned above, the rigorous theory of valuation in presence of all such effects can be quite challenging, leading to models that are based on advanced mathematical tools such as semilinear PDEs or BSDEs, which make numerical analysis difficult and slow. See for example \citet{Elkaroui} for an example of how asymmetric interest rates, even in absence of credit risk, lead to BSDEs. The papers
\citet{BFP15} and \citet{bichuch} deal with the mathematical analysis of valuation equations in presence of all the abovementioned effects and risks, except KVA. \citet{biffis} analyzes such effects in the area of life insurance contracts, and longevity swaps in particular.

Isolating and computing each individual adjustment is difficult because there is a marked interplay between them in pricing. Therefore, the causes of these adjustments are accounted for at the level of the contract payoffs and the resulting all-inclusive price is written as a solution to an advanced mathematical problem of the type mentioned above. Is there a
case, even for a simple contract, where this all-inclusive price of an uncollateralized contract can be calculated analytically? We present here an answer in the affirmative.

More specifically, we show that for standard benchmark products the above mathematically challenging structures can be solved analytically under a few simplifying assumptions. The solution is expressed in terms of the same explicit formula used for standard derivatives in the absence of these adjustments,
namely the Black-Scholes pricing formula for vulnerable options on dividend-paying assets. This leads to a
closed-form solution for the all-inclusive price of a benchmark product, the vulnerable call option, which then enables an analysis of all such effects that is more approachable from a numerical point of view. This link with the all-familiar Black-Scholes formula may be a way to reach out to a large portion of market participants and traders that are often discouraged by the full mathematical complexity of nonlinear valuation and the related nonlinearity valuation adjustment (NVA), see \citet{nva}.

The paper is structured as follows.

Section \ref{sec2} investigates the replication of a defaultable bond using CDS contracts
written on the same name as the bond, with special emphasis on the assumptions
required for the default time. In Section \ref{sec3}, the replication approach is used to derive the PDE satisfied by the pre-default pricing function. We show that this PDE is equivalent to that obtained in \citet{BJR2005} using the martingale measure approach. The solution to this PDE is then expressed as the Black-Scholes price when the underlying stock pays dividends, but with appropriately chosen parameters reflecting the hedger's funding costs. In Section \ref{sec4}, the adjusted cash flow approach developed in \citet{PPB11},  \citet{BP14} and \citet{BFP15} 
is used to obtain the price of the same option via an expected value of adjusted discounted cash flows; we show that it leads to exactly the same expression in terms of the Black-Scholes formula with dividends as the replication/PDE approach of Section \ref{sec3}. The appeal of the Black-Scholes formula with dividends is its tractability, enabling the sensitivity analysis of the price in terms of funding and repo rates and credit spread, as outlined in Section \ref{sec5}.
 \section{Replication of a Defaultable Bond using CDS Contracts} \label{sec2}

In this preliminary section, we discuss the issue of valuation of a defaultable bond in the simple model
with funding account and traded CDSs. Special emphasis is put on
the mathematical assumptions underpinning commonly used replication arguments,
assumptions that are frequently ignored in the existing literature.

\ssc{Dynamics of the Defaultable Bond Price}

Assume that we want to replicate a zero-recovery defaultable bond in a
financial market with an unsecured funding account with rate
$f_t$  dubbed the treasury rate and a market CDS, which is traded at null price, on the company that issued the bond.
The premium leg the CDS is assumed to pay a constant, continuous in time market spread
$r^{CDS}$ and the protection leg pays one at the default of the bond
and nothing otherwise. Recall that the market spread is computed by equating the value
of the protection leg with the value of the premium leg.
As we shall show in Section \ref{default}, the present postulates regarding the
market spread may only hold under specific assumptions on the probability distribution of the default time under
the real-world probability.

The price process $B$ of the zero-recovery defaultable bond maturing at $T$ is given in terms of
the point process $J$, which jumps to one when default occurs and stays zero otherwise.
Specifically, we have
\[
B_t= \I_{\{J_t=0\}} \tilde{B}_t = \I_{\{\tau>t\}}\tilde{B}_t
\]
where the yet unspecified process $\tilde{B}$ represents the pre-default price of the bond.

We will now provide intuitive replication arguments leading to the dynamics of the bond price;
a more formal derivation is postponed to the next subsection. We assume here that there has been no default yet,
but it may happen with a positive probability between
the dates $t$ and $t+dt$ for an arbitrarily small time increment $dt$.
To show how to replicate a long position in the defaultable bond,
let us consider the transactions an investor enters into at time $t<\tau\wedge T$:
\begin{enumerate}
	\item borrow $\tilde{B}_t$ from the treasury and use it to buy one defaultable bond;
	\item buy a number $\tilde{B}_t$ of CDS contracts on the same name.
\end{enumerate}
We have established a long position in the defaultable bond, and
everything else forms the reverse of the replicating portfolio. Hence, formally, the replicating
portfolio consists of the short position in the CDS and the long position in the treasury.

We now look at investor's portfolio at time $t+dt$:
\begin{enumerate}
	\setcounter{enumi}{2}
	\item if there is a default (i.e., $J_{t+dt}=1$), then each of the $\tilde{B}_t$
	CDS contracts pays 1;
	\item  if there is no default (i.e., $J_{t+dt}=0$), then he sells the bond for
	$\tilde{B}_{t+dt}$;
	\item either way, he pays the premium leg $r^{CDS} dt$ for each of the $\tilde{B}_t$ CDS contracts and
	pays back the loan to the treasury, which amounts to $\tilde{B}_t(1+f_t\, dt)$.
\end{enumerate}
The overall gain over the time interval $(t,t+dt)$ is
\[
\tilde{B}_t\I_{\{J_{t+dt}=1\}}+\tilde{B}_{t+dt}\I_{\{J_{t+dt}=0\}}
-\tilde{B}_t r^{CDS} dt-\tilde{B}_t(1+f_t\, dt) . 
\]
Equating this to zero to ensure replication and using
the fact that the first indicator above is just $dJ_t$ and that
we assumed $J_t=0$ (no default at time $t$), we obtain the dynamics for $B$
\be
dB_t-B_t(r^{CDS} + f_t)\,dt+B_t\, dJ_t=0 ,  \label{noarb}
\ee
and thus, since $B_T = \I_{\{\tau> T \}}$, we have for all $t \in [0,T]$
\be
B_t=\I_{\{\tau>t\}} e^{-\int_t^T(r^{CDS}+f_u)\, du} . \label{bondp}
\ee
Related dynamics were derived in \cite[Eq. (2.4)]{EJY2000} using different arguments and without addressing credit default swaps.
\ssc{Assumptions Underpinning Replication Arguments}  \label{default}

The above computations did not require the exact knowledge of a specific
distribution of a random time modelling the default event. Nevertheless, certain
conditions need in fact to be imposed on the default time for the replication argument to be valid.
To explain why additional assumptions are needed, let us denote the CDS price process by $S(\kappa )$ with $\kappa =r^{CDS}$
and let the treasury rate be a constant $f>0$.

\begin{pro}\label{distrib}
	The above replication of the defaultable bond holds whenever the probability distribution of $\tau $ is continuous and its support includes $[0,T]$.
\end{pro}

\begin{proof}
	According to our assumptions, the process $B^f$ satisfies $B^f_t = e^{ft}$ and the CDS price process
	jumps from zero to one at default and afterwards grows at the treasury rate
	\be \lab{cds1}
	S(\kappa ) = \I_{\{ t \geq \tau \}} e^{f(t- \tau )}.
	\ee
	An essential assumption in this step is that the fair spread is constant. Since we should ensure that the
	model with two assets, $B^f$ and $S(\kappa )$, is arbitrage-free, we postulate that there exists a probability
	measure $\mathbb{Q}$, equivalent to the real-world probability and such that $S_t(\kappa )$ is computed using the risk-neutral
	valuation under $\mathbb{Q}$. Standard computations show that, for a fixed spread $\kappa $, the equality $S_t(\kappa )=0$ may hold
	before default if and only if the intensity of default under $\mathbb{Q}$ is constant on $[0,T]$
	(see, for instance, Section 2.4.2 in \citet{BJR2009} or equation \eqref{ekappa}). Consequently, the probability distribution
	of default time under the real-world probability is continuous with a positive density on $[0,T]$, so that this interval
	is included in the support of distribution of $\tau $.
	
	In the second step, we will show that the postulate that the interval $[0,T]$ is included in the support of the
	real-world probability distribution of $\tau $ is also required for the replication argument to be strict.
	To this end, let us consider a self-financing trading strategy $\phi = (\phi^1, \phi^2)$ in assets $S(\kappa )$ and $B^f$,
	which is stopped at time $\tau \wedge T$ and replicates the defaultable bond, meaning that
	$V_{\tau \wedge T} (\phi ) = B_{\tau \wedge T} = \I_{\{ \tau > T\}}$.  A strategy $\phi = (\phi^1, \phi^2)$ is self-financing if its value
	\be \lab{cds2}
	V_t(\phi ) := \phi^1_t S_t(\kappa ) +  \phi^2_t B^f_t
	\ee
	satisfies
	\[
	dV_t(\phi )  = \phi^1_t \big(  dS_t(\kappa ) - \kappa \, d(t \wedge \tau ) \big) + \phi^2_t \, dB^f_t
	\]
	since the CDS pays negative dividends at the constant rate $r^{CDS} = \kappa $ up to time $\tau $.
	For the pre-default value of $\phi $, denoted as $\wt V(\phi )$, from \eqref{cds1} and \eqref{cds2}, we obtain $\wt V_t(\phi ) = \phi^2_t B^f_t $.
	Therefore,
	\be \lab{cond1b}
	dV_t(\phi )  = - \kappa \phi^1_t \, dt + \phi^2_t \, dB^f_t =  - \kappa \phi^1_t \, dt + f \wt V_t(\phi ) \, dt .
	\ee
	It is worth noting that our computations so far do not depend on the contingent claim we aim to replicate.
	
	In the last step, we specialize our trading strategy to zero-recovery bond by postulating that $V_{\tau } (\phi )= 0, \P$-a.s.
	This entails the jump of $V(\phi )$ at $\tau $ satisfies
	\begin{eqnarray*}
		\Delta_{\tau } V(\phi ) = V_{\tau } (\phi )-V_{\tau -}(\phi )= - V_{\tau -} (\phi )
		= - \wt V_{\tau -} (\phi ) 
		= \phi^1_{\tau } \Delta_{\tau } S(\kappa ) =  \phi^1_{\tau },
	\end{eqnarray*}
	which in turn leads to the following condition (which formally holds a.e. with respect to the Lebesgue measure)
	\be  \label{cond2a}
	\phi^1_t = - \wt V_{t-} (\phi ), \quad \forall t \in [0,T].
	\ee
	Then the property that condition  \eqref{cond2a} holds a.e. is equivalent to the equality $\phi^1_{\tau }=- \wt V_{\tau } (\phi )$, provided
	that the distribution of $\tau $ under $\P$ is continuous and $[0,T]$ is included in its support. Then the replicating strategy is
	independent of a particular distribution of $\tau $ satisfying these conditions. In other words, the exact knowledge of this distribution
	is immaterial for the problem at hand.
	
	Let us stress that equality \eqref{cond2a} should not be postulated a priori when allowing for more general distributions of default time.
	For instance, when $\P ( \tau \in (t_1,t_2 ) ) =0$, then we should set $\phi^1_t = 0$ for all $t \in (t_1,t_2)$ (see also the analysis in \citet{R1999}
	for the discontinuous case).
	
	We are now in a position to derive explicitly the replicating strategy.  By combining  \eqref{cond1b} with \eqref{cond2a}, we obtain
	\[
	d\wt V_t(\phi ) = (\kappa + f) \wt V_t(\phi ) \, dt
	\]
	with the terminal condition $\wt V_T(\phi )= 1$. We recover \eqref{noarb}
	\[
	\wt V_t (\phi ) = e^{- (\kappa + f)(T-t)} = \wt B_t
	\]
	and so
	\[
	B_t = \I_{\{ \tau >t \}} \, e^{- (\kappa + f)(T-t)} = \I_{\{ \tau >t \}}\, e^{- (r^{CDS} + f)(T-t)} .
	\]
	We have thus shown that the strategy
	\[
	\phi^1_t = -  e^{- (\kappa + f)(T-t)}, \quad \phi^2_t = (B^f_t)^{-1} e^{- (\kappa + f)(T-t)}
	\]
	is self-financing and replicates the defaultable bond $B$ on $[0, \tau \wedge T]$.
\end{proof}

\ssc{No-arbitrage and the Martingale Method} \label{marti}

An arbitrage-free pricing model for the CDS and the defaultable bond can also be constructed using directly
the so-called martingale mathod. In this modeling approach, one may take {\bf any} probability measure $\Q$ equivalent to $\P$ as a
{\bf postulated} martingale measure. To identify a martingale measure in our set-up, we require that $\Q$ should be consistent with the assumed properties of the CDS:  the spread equals $\kappa(t,T) >0 $ and the CDS pays one unit of cash at the moment of default,
provided that the default event occurs prior to or at $T$. As usual, the market CDS should have the value equal to zero
at any time before the default event.

Simple computations show that, in general, the {\bf market} (or {\bf fair}) spread of the CDS can be computed from the following formula
when the interest rate $f$ is constant
\begin{equation} \label{ekappa}
\kappa (t,T) = - \frac{\int_{(t,T]} e^{-fu}\, dG(u)}{ \int_{(t,T]} e^{-fu} G(u) \, du}
\end{equation}
where $G(t) := \Q(\tau >t)$.  As already mentioned,  it is now possible to prove that
the necessary and sufficient condition for the possibility of having a constant positive fair CDS rate $\kappa (t,T) = \kappa >0 $
is that the distribution of $\tau $ under $\P$ is continuous and has the support $[0,T]$,
exactly what was assumed in Proposition \ref{distrib}. In essence, this is due to the fact that for any positive density function on $[0,T]$ there exists a unique probability measure $\Q$ equivalent to $\P$ such that the distribution of $\tau $ under $\Q$ is exponential with parameter $\lambda = \kappa $ where a constant $\kappa >0$ is given in advance, provided that the interest rate $f$ is constant (notwithstanding the level of $f$). Then $\kappa (t,T) = \kappa $ and $\lambda$ has a natural interpretation as the default intensity under the CDS pricing measure $\Q$. The converse implication is valid as well.

Since we may show that the martingale measure $\Q$ is unique on ${\cal H}_T$,  where  ${\cal H}_t :=\sigma(\I_{\{\tau\leq u\}}, u\leq t)$ is the filtration generated by the default process,  the model with two assets, the funding account and the CDS, is complete.
Hence, from the Fundamental Theorem of Asset Pricing, any contingent claim $X$ maturing at $T$ can be replicated and its price,
which is defined as the value of a replicating strategy, can be also computed using the following version of the classical risk-neutral valuation formula
\be \label{eqrnv}
V_t = B^f_t \, \EQ \left( \frac{X}{B^f_T} \, \Big| \, \H_t \right).
\ee
For a claim with zero recovery, we obtain
\[
V_t = B^f_t \, \EQ \left( \frac{X \I_{\{\tau > T\}}}{B^f_T} \, \Big| \, \H_t \right),
\]
which reduces to
\[
V_t = \I_{\{ \tau >t \}}\, B^f_t (G_t)^{-1} \, \EQ \left( \frac{X G_T }{B^f_T} \right)
\]
where
\[
G_t = \Q ( \tau > t ) = e^{-\lambda t} = e^{-\kappa t}.
\]
The defaultable bond corresponds to $X=1$ and thus
\[
B_t = \I_{\{ \tau >t \}}\, B^f_t (G_t)^{-1} \frac{G_T }{B^f_T} = \I_{\{ \tau > t \}}\, e^{- (\kappa + f)(T-t)} .
\]
One may observe that the only claims with zero recovery in this model are zero-recovery bonds with differing, but constant, face values. Obviously, the valuation problem for claims with non-zero recovery can also be solved using \eqref{eqrnv}.

\sect{Vulnerable Call Option Pricing by Replication}
\label{sec3}

After a detailed analysis of valuation of the zero-recovery defaultable bond, we will now address a more advanced
problem of valuation of vulnerable options on some risky asset. Once again, our goal is to compare various approaches
and to identify the underlying assumptions, which are frequently ignored in the existing literature.

Denote by $\mathbb{F}= ({\cal F}_t)$ where ${\cal F}_t  :=\sigma(S_u, u\leq t)$ the natural filtration
generated by the price process of a traded asset (stock). Let the maturity date $T$ be fixed
and let $X$ be an ${\cal F}_T$-measurable integrable random variable.
Assume that the default time $\tau$ is a positive random variable on the
probability space $(\Omega, {\cal F},\P)$. The default time generates a
filtration $\hh = ({\cal H}_t)$ where  ${\cal H}_t :=\sigma(\I_{\{\tau\leq u\}}, u\leq t)$, which is used to progressively enlarge
$\ff $ in order to obtain the full filtration $\gg =({\cal G}_t)$ where ${\cal G}_t :={\cal F}_t\vee {\cal H}_t$.
We work under the assumption that $F_t :=\P (\tau\leq t\,|\,{\cal F}_t)$ is a continuous, increasing
function and $F_t<1$ for any $t$. Note that this assumption on the default time
has already appeared in \citet{EJY2000} in conjunction with the hypothesis (H) and is in line with what was assumed in Proposition \ref{distrib}.

Let $A$ be a {\bf contract} (vulnerable call option) that costs $P_0$ at time $0$ and has the payoff $X$ at maturity time $T$ where
\[
X=\I_{\{\tau>T\}}(S_T-K)^+ .
\]
Here $\tau$ is interpreted as the default time of the counterparty to the contract, that is, the issuer of the option.
We wish to find the {\bf price} $P_t$, $t\in[0,T]$, of this contract
for an investor who replicates a long position using financial instruments available in the market.

We now consider a market with the following {\bf primary assets} $(A^1,A^2,A^3,A^4)$:
\begin{enumerate}[i)]
	\item an unsecured funding account with the interest rate $f$;
	\item a stock (the underlying asset of the contract);
	\item a repo agreement on the stock with the repo rate ${h}$;
	\item a zero-recovery defaultable bond with  the rate of return $r^C$ issued by the counterparty.
\end{enumerate}
At time $t$, the price $P_t^i$ of the asset $A^i$ is given by
\[
P_t^1=B_t^f,\;\; P_t^2=S_t,\;\; P_t^3=0,\;\; P_t^4=B_t
\]
and the gains process since inception of $A^i$ is denoted by $G_t^i$ with $G^i_0=0$ for all $i$.

As a preliminary step, we specify the model inputs: the treasury rate $f$, the repo rate $h$ and the bond rate of return $r^C$.
Note that the rates $f,h$ and $r^C$ are postulated to be constant (or, at least, deterministic) and they are known.
We assume also that the process $S$ is continuous (obviously, $B^f$ is continuous as well).
We will later assume, in addition, that the stock price volatility $\sigma $ is known as well.
Hence we seek for the pricing formula in terms of the model parameters $f,h,r^C$ and $\sigma $  and
the option data: $T$ and $K$.

Note that, in principle, all these quantities are observed in the market, provided that the volatility is understood as the implied volatility.
By contrast, we do not need to assume that the CDS on the counterparty is traded, although this postulate would not change
our derivation of the option pricing formula, and the knowledge of the CDS spread $r^{CDS}$ is immaterial. In fact, we know from the
preceding section that, for a fixed level of the treasury rate $f$, there is one-to-one correspondence between $r^C$ and $r^{CDS}$.

Let us now determine the gains processes.
Buying one repo contract amounts to selling the shares of stock against
cash, under the agreement of repurchasing them back at the higher price
that includes the interest payments corresponding to the repo rate.
(Selling the repo results in the opposite cash flows.)
Any appreciation (or depreciation) in the stock price is part of the positive (or negative) gains,
while the outgoing repo interest payments are negative gains: $dG_t^3= dS_t-{h}S_t\, dt$.

Under the standing assumption that the pre-default rate of return $r^C$ on the counterparty's bond is deterministic, we obtain
\[
B_t = \I_{\{\tau > t\}} \, e^{- \int_t^T r^C_u \,  du} =  (1- J_t)e^{- \int_t^T r^C_u \,  du}
\]
where $J_t :=  \I_{\{\tau \leq t\}}$ is the point process that models the jump to default of the counterparty. The gains have negative terms
for outgoing cash flows corresponding to the drop in the bond value at the time of default. To summarize,
the gains of primary assets are given by
\begin{eqnarray} \label{gainn}
dG_t^1=fB_t^fdt,\;\; dG_t^2= dS_t,\;\; dG_t^3= dS_t-{h}S_t\, dt,\;\;  
dG_t^4= r_t^CB_t\, dt-B_{t-} \, dJ_t .\nonumber
\end{eqnarray}

A {\bf trading strategy} $\varphi =(\varphi^1,\varphi^2,\varphi^3,\varphi^4)$
gives the number of units of each primary asset purchased to build a portfolio.
Let $\beta\in[0,1]$ be a constant.  A trading strategy $\varphi$ is {\bf admissible} if at any date $t$
the investor can only use the repo market for a fraction $\beta$ of the stock amount required
and the rest has to be obtained in the stock market with funding from the treasury.
The {\bf wealth} at time $t\in[0,T]$ of the portfolio resulted from an admissible strategy $\varphi$ is denoted by $V_t^\varphi$ and equals
\[
V_t^\varphi= \sum_{i=1}^4 \varphi_t^i P_t^i
\]
and the gains process associated with this strategy satisfies $G_0^\varphi=0$ and
\begin{equation}
dG_t^\varphi:=\sum_{i=1}^4 \varphi_t^i\, dG_t^i .\label{gains_strategy}
\end{equation}
We then say that a strategy $\varphi$ is {\bf self-financing} if for all $t \in [0,T]$
\begin{equation}
V_t^\varphi=V_0^\varphi+G_t^\varphi.  \label{sf}
\end{equation}
An admissible trading strategy $\varphi$ {\bf replicates} the payoff of a
contract $A$ if $V_T^\varphi=X$. We define the time $t$ price of a contract $A$ as the wealth
$V_t^\varphi$ of the portfolio corresponding to the replicating strategy
\be
P_t:=V_t^\varphi. \label{price_t}
\ee
The existence of the specific primary assets in our market ensures that
any claim is {\bf attainable}. In fact, the market under study is complete and no-arbitrage arguments show that the price of any contract is unique.

Recall that the replicating portfolio is the negative
of the hedging portfolio that an investor holding the contract $A$ would
build to protect against market and counterparty risks.
In other words, the replicating strategy replicates not only the payoff of the option, but also
the market risk and the credit risk profiles of a long position in the option.
At date $t$ before default,  the investor builds a replicating portfolio for a long position in
the option knowing that the assumptions on $\tau$ imply that default may occur
between $t$ and $t+dt$ for an arbitrarily small $dt$. To replicate the contract the investor:
\begin{enumerate}
	\item buys $\beta\Delta_t$ repos, borrows $\beta \Delta_t S_t$ from treasury to buy and deliver $\beta\Delta_t$ shares, and receives
	$\beta \Delta_t S_t$ cash which is paid back to treasury;
	\item borrows $(1-\beta)\Delta_tS_t$  from treasury and buys $(1-\beta)\Delta_t$ shares;
	\item buys $P_t/B_t$ units of the counterparty bond in order to match the value of this portfolio and the option payoff.
\end{enumerate}
Of course, at this moment the option price $P_t$ is yet unknown, but it will be found from the matching condition \eqref{price_t} combined
with the terminal payoff $X$.  This replicating portfolio produces the following admissible strategy
\begin{equation}
\theta_t :=\Bigg{(}-\frac{
	(1-\beta)\Delta_tS_t}{B_t^f},(1-\beta)\Delta_t,\beta\Delta_t,\frac{P_t}{B_t}\Bigg{)}.
\label{strategy}
\end{equation}

At time $t+dt$ the investor:
\begin{enumerate} \setcounter{enumi}{3}
	\item receives $\beta\Delta_t$ shares from repo and sells them for $\beta\Delta_tS_{t+dt}$;
	\item borrows from treasury $\beta\Delta_tS_{t}(1+hdt)$ to close the repo;
	\item sells $(1-\beta)\Delta_t$ shares for $(1-\beta)\Delta_tS_{t+dt}$;
	\item sells the counterparty's bond for $P_t B_{t+dt}/B_t$;
	\item pays back to the treasury the amount $(1-\beta)\Delta_tS_t(1+fdt)$.
\end{enumerate}
The change in the wealth of the replicating position resulting from these transactions equals
\begin{align*}
V_{t+dt}^\theta-V_t^\theta
&=\beta\Delta_tS_{t+dt}-\beta\Delta_tS_t(1+hdt)+(1-\beta)\Delta_tS_{t+dt}
+\frac{P_t}{B_t}\, dB_t- (1-\beta)\Delta_tS_{t}(1+fdt)\nonumber\\
&=\beta\Delta_t\,dS_t-\beta h \Delta_t S_t\, dt+(1-\beta)\Delta_t\, dS_t
+\frac{P_t}{B_t}\, dB_t
-(1-\beta)f\Delta_tS_t\, dt\nonumber\\
&=\Delta_t\, dS_t - \big( (1-\beta)f + \beta h \big) \Delta_tS_t\, dt+P_t(r^Cdt-dJ_t).
\end{align*}
This can be derived formally by using \eqref{gainn} and computing the gains process (\ref{gains_strategy}) associated with the portfolio $\theta$ given by \eqref{strategy}
\begin{align*}
dG_t^\theta&= -\frac{(1-\beta)\Delta_t S_t}{B_t^f}fB_t^f dt
+(1-\beta)\Delta_t \, dS_t
+\beta\Delta_t (dS_t-hS_t\, dt)
+\frac{P_{t}}{B_t} (r_t^CB_t\, dt- B_{t-}\, dJ_t)\nonumber\\
&=\Delta_t\, dS_t- \big( (1-\beta)f+\beta h \big) \Delta_tS_t\, dt+P_t (r_t^C dt-dJ_t)
\label{gains_theta}
\end{align*}
where we used the equality $B_{t-}=B_t$, which holds before default. Note also that the wealth of $\theta $ at default equals zero, which is consistent
with the option payoff at default. Consequently, we may and do set $\theta_t = (0,0,0,0)$ for $t> \tau $.

Let us now focus on the pricing problem before default.
Since $dV_t^\theta=dG^\theta_t$ (from \eqref{sf}) and $dP_t=dV_t^\theta$ (from \eqref{price_t}), we have
\begin{equation}
dP_t= \Delta_t\,dS_t - \big( (1-\beta)f+\beta h \big) \Delta_tS_t\,dt+P_t(r_t^Cdt-dJ_t). \label{price_eq}
\end{equation}
To derive the pre-default pricing PDE, we assume that under the statistical probability $\P$ the stock price is governed by
\[
dS_t = \mu_t S_t \, dt + \sigma S_t \, dW_t
\]
and the price $P_t$ can be expressed as
\[
P_t= \I_{\{\tau>t\}} \tilde P_t = \I_{\{\tau>t\}} v(t,S_t) = (1- J_t) v(t,S_t)
\]
for some function $v(t,s)$ of class ${\cal C}^{1,2}$. Then the Ito formula yields
\begin{eqnarray*}
	dP_t=(1-J_t)\, dv (t,S_t)+ v(t,S_t) \, d(1-J_t) 
	= (1-J_t)\, dv (t,S_t)- v(t,S_t) \, dJ_t
\end{eqnarray*}
and
\begin{eqnarray}
dP_t =(1-J_t)\Big{(}v_t(t,S_t) +\frac{\sigma^2S_t^2}{2}v_{ss}(t,S_t) \Big{)}dt 
+(1-J_t)v_s(t,S_t)\, dS_t-v(t,S_t)\, dJ_t.
\label{product1}
\end{eqnarray}
By equating the $dS_t, dt$ and $dJ_t$ terms in \eqref{price_eq} and \eqref{product1}, we obtain the following equalities in which
the variables $(t,S_t)$ were suppressed
\begin{align}
\Delta_t&=(1-J_t) v_s, \nonumber \\
(1-J_t)&\Big{(}  v_t+\frac{\sigma^2S_t^2}{2}\, v_{ss} + \big((1-\beta)f + \beta h\big) S_t v_s\Big{)}
-(1-J_t)r_t^C v=0, \label{PDE1}\\
-P_{t}\, dJ_t&=-v \, dJ_t. \nonumber
\end{align}
The pre-default pricing PDE for the function $v(t,s)$ is now obtained from \eqref{PDE1} as
\begin{equation}
v_t+ \big( (1- \beta )f +\beta h \big) s\, \frac{\partial v}{\partial s}+\frac{\sigma^2 s^2}{2} \frac{\partial v^2}{\partial s^2}- r_t^C v = 0 \label{PDE}
\end{equation}
with terminal condition $v(T,s)=(s-K)^+$. One recognizes \eqref{PDE} as the Black-Scholes PDE when the underlying stock pays dividends.
To see this, it suffices to take the discount rate to be the return on the defaultable bond
$r:=r^C$ and the instantaneous dividend yield to be the bond spread over the effective funding rate:
$q:=r^C - f^{\beta }$  where the {\bf effective funding rate} is defined as the weighted average: $f^{\beta }:= (1-\beta)f + \beta h$.
We conclude that the following result is valid.

\begin{pro} \label{prirep}
	The time $t$ price of the vulnerable call option obtained by replication equals
	\be
	P_t=\I_{\{\tau>t\}}\Big{(}S_t e^{-q(T-t)}N(d^q_1)-Ke^{-r^C(T-t)}N(d^q_2)\Big{)}
	\label{final_price}
	\ee
	with $q=r^C - f^{\beta}$ and
	\[
	d^q_1= \frac{\log\frac{S_t}{K}+(r^C-q+\frac{\sigma^2}{2})(T-t)}{\sigma\sqrt{T-t}},
	\;\; d^q_2=d^q_1-\sigma\sqrt{T-t}.
	\]
\end{pro}

It is worth noting that \eqref{final_price} may also be derived from \eqref{price_eq} through probabilistic means without resorting to the pricing PDE.
From  \eqref{price_eq}, we obtain the following equation for the pre-default price $\tilde P$
\begin{equation}
d\tilde P_t=-  f^{\beta } \Delta_tS_t\,dt+\Delta_t\,dS_t+ r_t^C \tilde P_t\, dt . \label{prce_eq}
\end{equation}
Let now $\Qhf$ be the probability measure, which is equivalent to $\P$, and such that the drift of the risky asset $S$ under $\Qhf$ is equal to the effective funding
rate $f^{\beta }$. Then the process $\tilde P$ is governed under $\Qhf$ by
\[
d\tilde P_t -r _t^C \tilde P_t\, dt = \Delta_t \sigma S_t \,dW^{\beta}_t
\]
with terminal condition $\tilde P_T = (S_T-K)^+$ where $W^{\beta}$ is the Brownian motion under  $\Qhf$.
This leads to the following probabilistic representation for $\tilde P_t$
\begin{eqnarray*}
	\tilde{P}_t&=&e^{-r^C (T-t)} \, \Ehf [(S_T-K)^+\,|\,{\cal F}_t]\\
	&=&e^{-(r^C - f^{\beta } )(T-t)} \, \Ehf[e^{- f^{\beta }(T-t)}(S_T-K)^+\,|\,{\cal F}_t],
\end{eqnarray*}
which in turn yields  \eqref{final_price} through either standard computations of conditional expectation or by simply noting that it is given
by the Black-Scholes formula with  the interest rate $f^{\beta }$ and no dividends.

\begin{rems}  \label{remar1} {\rm
		i) If we model the defaultable bond as in \eqref{bondp} with
		$r=r_t^C=r^{CDS}+f$ where the CDS spread $r^{CDS}$ (rather than the bond return $r^C$) is taken as a model's input, then the pricing equation \eqref{final_price}
		holds with $q :=r^{CDS}-\beta (h-f)$. In other words, the option pricing formula \eqref{final_price} is still valid when the defaultable bond
		is replaced by the counterparty's CDS in our trading model.
		
		\noindent  ii)
		PDE \eqref{PDE} is in fact equivalent to PDE (32) obtained in \cite[Eq. 4.4]{BJR2005}
		using the martingale approach. To see this, it suffices to rewrite \eqref{PDE} with the dynamics of the primary assets
		\begin{align*}
		&dB_t^f =fB_t^fdt, \\
		&dS_t =\mu S_t\,dt+\sigma S_t\,dW_t , \\
		&dB_t =B_{t-}(\mu_3\,dt-dM_t) =B_{t-} \big( (\mu_3+\xi_t)dt-dJ_t \big),
		\end{align*}
		where $\mu_3=f$ and $M_t := J_t + \log{G_{t\wedge\tau}} = J_t - \xi_t$ where  $G_t:=\P (\tau>t|{\cal F}_t)$. The process $M$
		is commonly known as the compensated $\gg$-martingale of the default process $J$.
		
		\noindent  iii) Though the stock $S$ was assumed to pay no dividends, the present framework can be
		easily extended to the dividend-paying case. As a result, the effective funding rate $f^{\beta }$
		should be replaced by $f^{\beta } - \delta $.
	}
\end{rems}

\sect{Vulnerable Call Option Pricing by Adjusted Cash Flows Approach}
\label{sec4}

Let us consider again the problem of pricing the same vulnerable option, but
this time using the {\bf adjusted cash flows approach} originated in \citet{PPB11},
derived rigorously in \citet{BFP15} and presented in a wider context in
\citet{BP14}. We do not make here an attempt to justify their approach,  but we start instead with
the pricing equation (11) of \citet{BP14} and adapt it to
the present context of a vulnerable call option, which is an uncollateralized contract.
Note that the variation margin is $M$, while $N^C$ and $N^I$ are
the initial margin accounts for the two counterparties, resulting in the total
collateral account $C=M+N^C+N^I$. In our case, this means that all terms appearing in the last two
lines of equation (11)  in \citet{BP14} vanish. Moreover, the cash flow at default equals
zero (due to zero recovery convention for the vulnerable option) and the promised cash flow over time period $(t,t+dt)$ is
\[
\Pi(t,t+dt)=(S_T-K)^+\I_{\{t=T\}}.
\]
Hence the pricing equation (11) in \citet{PPB11}  reduces to
\begin{equation} \label{eqd}
V_t={E}^h \big[ \I_{\{\tau>T\}}D(t,T;f)(S_T-K)^+\,|\,{\cal G}_t \big]
\end{equation}
where $\Ehh$ is the expectation with respect to the probability measure $\Qhh$ that makes the drift
of the risky asset equal to $h$, meaning that
\[
dS_t = h S_t \, dt + \sigma S_t \, dW^h_t
\]
where $W^h$ is a Brownian motion under $\Qhh$. Furthermore, $\mathbb{G} = ({\cal G}_t)$ is the full filtration
that includes the information on default times and the discount factor $D(s,t;f)$ equals
\[
D(s,t;f) := \exp\left( -\int_s^t f_u \, du\right).
\]
We henceforth assume a constant treasury rate $f$ and we use the pre-default intensity
$\lambda$ under $\Qhh$ of the counterparty, which is defined in \cite[(40)]{BP14} by
\[
\I_{\{\tau>t\}}\lambda \, dt:=\Qhh ( \tau\in dt\,|\, \tau>t,{\cal F}_t ),
\]
to obtain the survival probability $G^h_t=e^{-\lambda t}$ where $G^h_t:=\Qhh(\tau>t\,|\,{\cal F}_t)$.
Note that this is consistent with the assumptions on $\tau$ in the replication
approach of Section \ref{sec3}. Using \eqref{eqd} and Cor. 3.1.1 of \citet{BJR2004}, we obtain
\[
{V}_t = \I_{\{\tau>t\}} (G^h_t)^{-1} \, \Ehh[D(t,T;f)(S_T-K)^+ G^h_T\,|\,{\cal F}_t].
\]
If $\tilde{V}$ denotes the $\ff$-adapted pre-default price process such that for all $t \in [0,T]$
\[
\I_{\{\tau>t\}}V_t=\I_{\{\tau>t\}}\tilde{V}_t,
\]
then from the above equation we immediately obtain
\[
\tilde{V}_t= (G^h_t)^{-1} \, \Ehh[D(t,T;f)(S_T-K)^+ G^h_T\,|\,{\cal F}_t] .
\]
Since $G^h$ is deterministic, for a constant treasure rate $f$, the pre-default price can be written as
\[
\tilde{V}_t=e^{-(\lambda+f)(T-t)} \, \Ehh[(S_T-K)^+\,|\,{\cal F}_t]
\]
or, equivalently,
\[
\tilde{V}_t = e^{-(\lambda+f-h)(T-t)} \, \Ehh[e^{-h(T-t)}(S_T-K)^+\,|\,{\cal F}_t].
\]
The last expectation can be computed yielding the usual Black-Scholes formula when the drift of the stock equals $h$
\[
\Ehh[e^{-h(T-t)}(S_T-K)^+)\,|\,{\cal F}_t]=S_tN(d_1)-Ke^{-h(T-t)}N(d_2)
\]
where
\[
d_1=\frac{\log{\frac{S_t}{K}+(h+\frac{\sigma^2}{2})(T-t)}}{\sigma\sqrt{T-t}}, \quad
d_2=d_1-\sigma\sqrt{T-t}.
\]
We conclude that the pre-default price process satisfies
\[
\tilde{V}_t= e^{-(\lambda+f-h)(T-t)} \big( S_t N(d_1)-Ke^{-h(T-t)}N(d_2) \big)
\]
and thus
\begin{equation}
V_t= \I_{\{\tau>t \}}  \Big( S_te^{-(\lambda+f-h)(T-t)}N(d_1)-Ke^{-(\lambda+f)(T-t)}N(d_2) \Big).
\label{master2}
\end{equation}
Upon setting $\lambda+f-h=q$ and $\lambda+f=r$, we deduce that \eqref{master2} coincides with
the pricing formula \eqref{final_price} obtained by replication for $\beta=1$. It is also not difficult to show
that $\lambda = r^{CDS}$ (in essence, this is due to the fact that the density of $\Qhh$ with respect to the martingale
measure  $\Q$ introduced in Section 2.2 is $\ff$-adapted). This shows that the adjusted cash flow method and
the replication approach lead to the same price for the vulnerable call option.

\sect{Sensitivity Analysis} \label{sec5}

In the final step, we perform the sensitivity analysis for the vulnerable call option by focussing on
the impact of the rates $f$ and $h$. We leave aside the parameter $r^C$, since in our model the investor
has the freedom to choose a particular combination of funding sources when purchasing shares, as formally
represented by the parameter $\beta \in [0,1]$, but the spread $r^C$ is given by the market and thus it is
natural to assume that it is fixed.


\begin{ex} \label{examp} {\rm
		Figure \ref{plot3d} shows the dependence of the pre-default price \eqref{master2} of a
		vulnerable call on the treasury rate $f$ and repo rate $h$ for $\beta=1$ when $S_t=80$, $K=100$, $\sigma=0.3$, $T-t=0.1$, $r^{CDS}=0.05$.
		The pre-default price of the vulnerable call is decreasing in the
		treasury rate $f$ and increasing in the repo rate $h$.}
	\begin{figure}[ht!]
		\centerline{\includegraphics[scale=0.6]{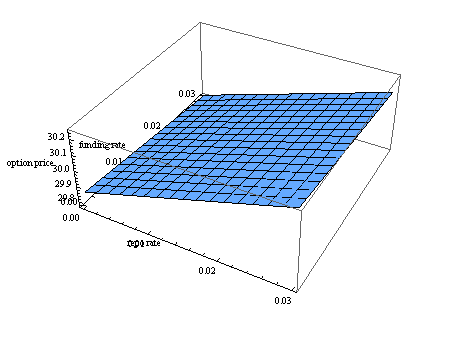}}
		\label{plot3d}
		\caption{
			Option price is increasing in repo rate $h$ and decreasing in funding rate $f$}
	\end{figure}
\end{ex}

To explain the dependence observed in Example \ref{examp} and perform a general sensitivity analysis, we first compute ``funding Greeks" when $\beta =1$, that is, all shares are purchased at repo. We obtain the following expressions
\begin{align}
&\partial_f \tilde{V}_t =\partial_{r^C}\tilde{V}_t=-(T-t)\tilde{V}_t < 0,  \label{pdf}\\
&\partial_h\tilde{V}_t=e^{-(h-f-r^{CDS} )(T-t)}(T-t)S_t N(d_1^q) > 0 , \label{pdh}
\end{align}
which means that the pre-default call price decreases in both the treasury rate $f$ and the bond return $r^C$,
but increases in the repo rate associated with the risky asset. Furthermore, the relative sensitivity to funding $\frac{\partial_f\tilde{V}_t}{\tilde{V}_t}=-(T-t)$
appears to be smaller in absolute value than the relative sensitivity to the repo rate $\frac{\partial_h\tilde{V}_t}{\tilde{V}_t}>T-t$.
This simple benchmark case highlights that the repo rate may have an important impact on the contract
value, often more significant than the treasury rate or the credit spread.

Let us now consider the price obtained in Section \ref{sec3} where the additional parameter $\beta\in[0,1]$ dictates
the structure of the funding arrangements for the investor. In view of \eqref{final_price} and Remark \ref{remar1} i), we obtain the following
funding Greeks:
\begin{align}
\partial_f \tilde{V}_t&=
-\beta (T-t)\tilde{V}_t 
+(1-\beta)(T-t)e^{(\beta (h-f)-r^{CDS})(T-t)}KN(d_2^q),\label{pdfb}\\
\partial_h\tilde{V}_t&=\beta e^{(\beta (h-f)-r^{CDS})(T-t)}(T-t)S_t N(d_1^q)\geq 0,  \nonumber
\end{align}
where the last inequality is strict when $\beta >0$.
In particular, for $\beta=1$ we recover \eqref{pdf}-\eqref{pdh} and for $\beta=0$ (pure treasury funding), we get
\begin{align*}
\partial_f \tilde{V}_t&=(T-t)e^{(f-r^{C})(T-t)}KN(d_2^q)>0, \\
\partial_h\tilde{V}_t&=0,
\end{align*}
where $f-r^C = -r^{CDS}<0$. In general, it is hard to determine the sign of the sensitivity $\partial_f  \tilde{V}_t$ given by \eqref{pdfb}, though it is clear that
it changes from a positive value for $\beta =0$ to a negative one for $\beta =1$.

To give an interpretation of funding Greeks, we observe that the contract's payoff can be written as $X = B_T (S_T-K)^+$, so it can be seen as a hybrid contract which combines the call option on the stock with the long position in the counterparty bond. For any $0< \beta \leq 1$ the
price $ \tilde{V}_t$ increases in $h$, since the cost of hedging the option component $(S_T-K)^+$ is manifestly increasing with $h$.

The price dependence on $f$ is a bit harder to analyze. Indeed, from representation \eqref{strategy} of the hedging portfolio, we see that for $0< \beta <1$  the dependence on $f$ is rather complex: the investor needs to borrow cash from $B^f$ (which grows at the rate $f$) and thus the cost of hedging increases in $f$, but he simultaneously invests in the bond $B$ (with the rate of return $r^C = f+r^{CDS}$ where $r^{CDS} $ is constant) so that the cost of hedging decreases in $f$. The net impact of both legs may be negative, in the sense that the price of the option decreases when $f$ increases. This is rather clear for $\beta =1$, since in that case the investor does not use $B^f$ for his hedging purposes (take $\beta=1$ in \eqref{strategy}) and we see that the cost of hedging the component $B_T$ in the payoff $X$ falls when $f$ increases. By contrast, when $\beta =0$ the value of $h$ is immaterial, and the increase of $f$ makes the option more expensive.
Finally, when only the CDS spread $r^{CDS} $ increases and $f,h$ are kept fixed, then the cost of hedging decreases as well, since the bond $B$ becomes cheaper.

\sect{Conclusions}  \label{sec6}

The mapping of the price computation to the Black-Scholes formula with dividends can be generalized to any local volatility model
(e.g., the displaced diffusion model), which would thus cover both an increasing and a decreasing volatility smile. Stochastic volatility models
would also be attractive to the industry, but any additional source of randomness would need to be hedged, thus requiring the
inclusion of additional hedging instruments to our market model and solving a suitable modification of the pricing PDE.
The local stochastic volatility models currently dominant in the industry would obviously pose the same problem.

In summary, we have shown that two alternative pricing approaches
lead to the same result in the benchmark case of a vulnerable call
option that includes funding, repo and credit risk. This confirms that
even in the presence of funding costs and repo contracts the martingale
method and the adjusted cash flow approach should be seen as
alternative tools facilitating the computation of the replication
price, rather than as alternative pricing paradigms. The reason for
this is that all these approaches are in fact either explicitly or
implicitly based on the concept of replication, as explained
more generally in \citet{BBFPR2015}. Furthermore, we show that the option price can be
expressed as a Black-Scholes formula with dividends, thus facilitating
the use of the funding Greeks in the valuation and sensitivity analyses.
In this context, we highlight the potentially important pricing impact of
the repo rate over the treasury rate and credit spread.
\bibliographystyle{plainnat}

\end{document}